\documentclass[]{article}

 	\usepackage[latin1]{inputenc}
 	\usepackage[T1]{fontenc}
	
	\usepackage{amsfonts}
 	\usepackage{amsmath}
 	\usepackage{amssymb}
  	\usepackage{amsthm}
  	\usepackage{color}
 	\usepackage[english]{babel}
	\usepackage{enumerate}
 	\usepackage{latexsym} 	
  	\usepackage[margin= 2.75cm]{geometry}
 	\usepackage{setspace}
 	\usepackage{url}

 	\onehalfspace
	\newcommand{\dsum}[3]{\displaystyle\sum_{{#1}}^{{#2}}{#3}}
	\newcommand{\set}[3]{\left\{{#3}\right\}_{{#1}}^{{#2}}}
	
	\newcommand{\corr}[2]{Corr\left({#1},{#2}\right)}
	\newcommand{\ccorr}[3]{Corr\left({#1}, {#2}\left|{#3}\right.\right)}
	\newcommand{\ecorr}[2]{Corr_{e}\left({#1},{#2}\right)}
	\newcommand{\dev}[2]{\delta\left({#1},{#2}\right)}
	\newcommand{\prob}[1]{p\left({#1}\right)}
	\newcommand{\altprob}[1]{p'\left({#1}\right)}
	\newcommand{\cprob}[2]{p\left({#1}\left|{#2}\right.\right)}

	\newtheorem{definition}{Definition}
	\newtheorem{lemma}{Lemma}
	\newtheorem{proposition}{Proposition}

\title{Generalised Reichenbachian Common Cause Systems}

\author{Claudio Mazzola}
\date{}

\begin{document}

\maketitle

\vspace{-32pt}
\begin{center}
\normalsize{School of Historical and Philosophical Inquiry \\ The University of Queensland \\ Forgan Smith Building (1), St. Lucia, QLD 4072, Australia \\ c.mazzola@uq.edu.au}
\end{center}


\section{Introduction}


Chancy coincidences happen everyday, but sometimes coincidences are just too striking, or too improbable, not to reveal the presence of some coordinating process. To wit, if all the electrical appliances in a building were to shut down at exactly the same time, it would not be unreasonable to search for a breakdown in their common power supply. Similarly, if the price of petrol were to simultaneously rise in all oil importing countries, it would be a fair bet that exporters had concertedly decided to reduce extraction. The principle of the common cause is the inferential rule governing instances of this kind: informally stated, it asserts that improbable coincidences are to be put down to the action of a common cause. 

Reichenbach \cite{reichenbach 1956} was the first to provide the principle of the common cause with a matematical characterisation. His treatment relied on three major ingredients: first, he represented improbable coincidences as positive probabilistic correlations between random events; second, he demanded that common causes should increase the probability of their effects; and third, he further required that conditioning on the presence, or on the absence, of a common cause should make its effects probabilistically independent from one another. 

Reichnenbach's treatment, however, was overly restrictive, as it rested on a too narrow conception of improbable coincidences, and on a correspondingly narrow understanding on the explanatory function of common causes. In \cite{mazzola 2013}, I accordingly proposed an improved interpretation of the principle, along with a suitably revided probabilistic model for common causes, which generalises Reichenbach's original formulation in two respects. On the one hand, it represents improbable coincidences not as positive correlations, but rather as positive differences between the correlation actually exhibited by a speficied pair of events, and the correlation that they should exhibit according to historical data, background beliefs, or established theory. On the other hand, and correspondingly, it demands that conditioning on the presence or on the absence of a common causes should restore the expected correlation between its effects. 

Reichenbach's understanding of the principle is demonstrably a special case of this interpretation, applying when the expected correlation between the events of interest is null. Nevertheless, there is one respect in which the probabilistic model proposed in \cite{mazzola 2013} is still not general enough. Like Reichenbach's original account, in fact, it depicts the action of a single common cause, and it is accordingly inadequate to capture instances whereby two coordinated effects are brought about by a system of distinct common causes. The aim of this paper is precisely to further expand the model in this direction. To this end, two avenues for the generalisation of the model will be explored, each based on a different probabilistic characterisation for systems of common causes. 

The article will be structured in three main sections. Firstly, in \S \ref{section: generalised common causes} the extended interpretation of the principle elaborated in \cite{mazzola 2013} will be briefly outlined, and given formal treatment. Next, in \S \ref{generalised hr-reichenbachian common cause systems} said interpretation will be incorporated into Hofer-Szab\'{o} and R\'{e}dei's Reichenbachian Common Cause Systems model \cite{hofer-szabo redei 2004}. Finally, in \S \ref{generalised m-reichenbachian common cause systems} the extended version of the principle will be integrated with my own revisitation of Reichenbachian Common Cause Systems \cite{mazzola 2012}. 


\section{Generalised conjunctive common causes}\label{section: generalised common causes}


Reichenbach originally applied the principle of the common cause to pairs of positively correlated, albeit causally unrelated, events. Before introducing his probabilistic model for common causes, a definition of probabilistic correlation is thus needed: 
\begin{definition}\label{definition: correlation}
Let $(\Omega, p)$ be a classical probability space with $\sigma$-algebra of random events $\Omega$ and probability measure $p$. For any $A,B,C\in \Omega$ such that $\prob{C}\neq 0$, we define:
\begin{equation}\label{equation: correlation}
\ccorr{A}{B}{C}:= \cprob{A\wedge B}{C} - \cprob{A}{C}\cprob{B}{C}.
\end{equation}
Moreover, 
\begin{equation}\label{equation: unconditional correlation}
\corr{A}{B} : = \ccorr{A}{B}{\cup_{X_{i}\in \Omega} X_{i}}.
\end{equation}
\end{definition} 
The expression $\ccorr{A}{B}{C}$ denotes the \emph{correlation} of events $A$ and $B$ \emph{conditional} on event $C$. The expression $\corr{A}{B}$, instead, denotes the \emph{absolute correlation} or \emph{unconditional correlation} of events $A$ and $B$. Two events are said to be \emph{positively} (\emph{negatively}) \emph {correlated} (conditional on another event) if their correlation (conditional on said event) is greater (smaller) than zero; by the same token, they are said to be \emph{uncorrelated} or \emph{probabilistically independent} (conditional on another event) if their correlation (conditional on that event) is equal to zero.

The existence of a positive correlation between two events is often an indication that one of them is a cause of the other. However, this is not invariably the case: as is well known, correlation does not imply causation. Reichenbach's interpretation of the common cause principle could indeed be seen as an attempt to preserve a one-one correspondence between probabilistic correlation and causal dependence \cite{hofer-szabo et al 2013}: in his account, acquiring information about the occurrence of a common cause should dissolve, as it were, any positive correlation between causally unrelated events. 
Reichenbach gave formal shape to this intuition by demanding that conditioning on the presence of a common cause, or on its absence, should make its effects probabilistically independent. The result was a probabilistic model for common causes known as \emph{conjunctive fork}. With only a slight terminological modification and few minor notational variants, we can introduce his model as follows:
\begin{definition}\label{definition: conjunctive fork}
Let $(\Omega, p)$ be a classical probability space with $\sigma$-algebra of random events $\Omega$ and probability measure $p$. For any three distinct $A, B, C\in \Omega$, the event $C$ is a conjunctive common cause for $\corr{A}{B}$ if and only if:
\begin{gather}
\prob{C} \neq 0  									\label{equation: c-fork 0} \\ 
\prob{\overline{C}} \neq 0							\label{equation: c-fork 0 complement}\\
\ccorr{A}{B}{C}  =  0					 		  	\label{equation: c-fork 1} \\
\ccorr{A}{B}{\overline{C}}  =  0							\label{equation: c-fork 2} \\
\cprob{A}{C} - \cprob{A}{\overline{C}}  >  0					\label{equation: c-fork 3} \\
\cprob{B}{C} - \cprob{B}{\overline{C}}  >  0.					\label{equation: c-fork 4} 
\end{gather}
\end{definition}

Conjunctive common causes, as just defined, are intended to explain the occurrence of non-causal positive correlations in two ways. On the one hand they increase the joint probability of their effects, consequently favouring their correlation, as established by the following proposition:
\begin{proposition}
Let $(\Omega, p)$ be a classical probability space with $\sigma$-algebra of random events $\Omega$ and probability measure $p$. For any three distinct $A, B, C\in \Omega$, if $C$ is a conjunctive common cause for $\corr{A}{B}$, then:
\begin{equation}\label{equation: positive correlation}
\corr{A}{B} > 0.
\end{equation}
\end{proposition}
On the other hand, conditions (\ref{equation: c-fork 1})-(\ref{equation: c-fork 2}) demand that the correlation between the effects of a conjunctive common cause should disappear when conditioning on the occurrence, or on the absence, of said cause: in jargon, we say that the common cause \emph{screens-off} the two effects from one another. This is meant to indicate that the positive correlation between the two effects is purely epiphenomenal, being a mere by-product of the underlying action of the common cause.

Reichebach's conjunctive common cause model has exerted considerable influence in both probabilistic causal modelling and the philosophy of science. To mention but few of its contributions to the latter field, it anticipated the probabilistic causality program \cite{good 1961,suppes 1970,cartwright 1979,skyrms 1980,eells 1991}, fostered the development of probabilistic accounts of scientific explanation \cite{salmon 1971,suppes zaniotti 1981}, and inspired causal interpretations of Bell's no-go theorem in quantum physics \cite{vanfraassen 1991}. Simultaneously, the Bayesian Networks movement in probabilistic causal modelling incorporated and generalised the screening-off constraints (\ref{equation: c-fork 1})-(\ref{equation: c-fork 2}) in the guise of the so-called Causal Markov Condition, according to which any two variables that are not related as cause and effect must be probabilisitically independent conditional on the set of their direct causes \cite{pearl 1988,pearl 2000, spirtes et al 2001}. Nonetheless, the conjunctive common cause model relies on a demonstrably restrictive understanding of the principle of the common cause, and on a correspondingly narrow conception of the explanatory function performed by common causes. 

To fully appreciate this, it will be instructive to start by taking a deeper look at the very thing the principle of the common cause is intended to apply to: improbable coincidences. Reichenbach, as we saw, understood improbable coincidences as positive correlations between causally unrelated events. Positively correlated events tend to be coinstantiated, so it is clear why positive correlations can be used to give coincidences a probabilistic representation. The problem is: in what sense, then, can coincidences between causally unrelated events be deemed \emph{improbable}? The underlying presupposition is that \emph{in general} causally unrelated events tend to be uncorrelated, so \emph{in general} positive correlations between such events are not to be expected. Reichenbach, in other words, applied the principle of the common cause to pairs of events that happen to be positively correlated, even though we would expect them to be not. To be even more explicit: he applied the principle to cases where the observed value of the correlation between two events is strictly higher than its expected value, \emph{which is zero}.  

Once the principle is presented in this way, however, it becomes apparent that there is no reason not to demand that it should equally apply to \emph{all} cases in which two events are more strongly correlated than expected, whatever the value of their expected correlation. The \emph{extended principle of the common cause} is specifically tailored to meet this demand. Compressed in one sentence, it claims that the role of common causes is to explain statistically significant deviations between the estimated value of a correlation and its expected value, by conditionally restoring the latter.

To illustrate, let us consider an economic example. Let us imagine that an econometric analysis revealed a strong positive correlation between holding a postgraduate degree and earning a higher-than-average income. This positive correlation, in and of itself, would not be surprising, as it would be consistent with both common sense and microeconomic theory: people who study more are likely to earn higher wages, owing to the comparatively scarce supply and higher productivity of skilled labour. But suppose that, in the case at hand, the estimated correlation were remarkably strong: strong enough to be significantly dissimilar from the average correlation reported by other similar studies. Then, excluding any mistakes in the analysis, it would be natural for one to wonder if there were anything about the selected sample, which could bring about said discrepancy. 

The extended principle of common cause urges that the explanation should be sought in the presence of some unacknowledged common cause. To wit, we may imagine that the econometric analysis in our example were conducted in relatively wealthy subpopulation. People coming from wealthy families are more likely to undergo additional years of study, since they can more easily afford the opportunity costs this involves; moreover, they are more likely to earn their degrees from renouned but expensive academic institutions, whose graduates have a higher chance to be hired in high-earning appointments. By simultaneously increasing the probability of holding a postgraduate degree and the probability of earning a higher-than-average income, family wealth would consequently increase their joint probability, and explain their stronger-than-usual correlation. 
 
Remarkably, in this case it would be unreasonable to require that the correlation between holding a postgraduate degree and of earning a higher-than-average income should disappear conditional on family wealth: after all, as we already noticed, some positive correlation between wage and qualification is to be expected. Rather, conditioning on the common cause should restore the expected correlation between the two events, consequently eliminating the apparent disagreement between the econometric analysis and the preceding studies. 

To provide the extended principle of the common cause with some formal bite, let us first define:
\begin{definition}\label{definition: deviation}
Let $(\Omega, p)$ be a classical probability space with $\sigma$-algebra of random events $\Omega$ and probability measure $p$. For any $A,B\in \Omega$, the \emph{deviation} of $\corr{A}{B}$ is the quantity
\begin{equation}\label{equation: deviation}
\dev{A}{B} := \corr{A}{B} - \ecorr{A}{B},
\end{equation}
where $\ecorr{A}{B}$ denotes the \emph{expected correlation} between $A$ and $B$.
\end{definition} 
Notice that the notions of \emph{deviation} and \emph{expectation}, as they are understood here, are not necessarily restricted to the corresponding statistical concepts: in particular, the expected correlation between two values may be determined by non-statistical means, e.g. on the basis of logical or mathematical rules, or simply on the basis of entrenched prior beliefs.  
On this basis, we can now define:
\begin{definition}\label{definition: generalized conjunctive fork}
Let $(\Omega, p)$ be a classical probability space with $\sigma$-algebra of random events $\Omega$ and probability measure $p$. For any three distinct $A, B, C\in \Omega$, the event $C$ is a \emph{generalised common cause} for $\dev{A}{B}$ if and only if:
\begin{gather}
\prob{C} \neq 0      							\label{equation: gc-fork 0} \\ 	
\prob{\overline{C}} \neq 0					\label{equation: c-fork 0 complement}\\
\ccorr{A}{B}{C} = \ecorr{A}{B}      				\label{equation: gc-fork 1} \\         
\ccorr{A}{B}{\overline {C}} = \ecorr{A}{B}	  		\label{equation: gc-fork 2} \\         
\cprob{A}{C} - \cprob{A}{\overline{C}} > 0      		\label{equation: gc-fork 3} \\         
\cprob{B}{C} - \cprob{B}{\overline{C}} > 0.     		\label{equation: gc-fork 4}   
\end{gather}
\end{definition}

Just like conjunctive common causes do for positive correlations, generalised common causes explain positive deviations in two ways. On the one hand, (\ref{equation: gc-fork 1})-(\ref{equation: gc-fork 2}) demand that conditioning on the presence of a common cause, or on its absence, should restore the expected correlation between its effects. On the other hand, 
generalised common causes increase the unconditional correlation between their effects, consequently generating the observed discrepancy between the estimated value of said correlation and its expected value:
\begin{proposition}
Let $(\Omega, p)$ be a classical probability space with $\sigma$-algebra of random events $\Omega$ and probability measure $p$. For any three distinct $A, B, C\in \Omega$, if $C$ is a generalised common cause for $\dev{A}{B}$ then
\begin{gather}\label{equation: positive deviation}
\dev{A}{B} > 0.
\end{gather}
\end{proposition}

Quite evidently, conjunctive common causes can be thought of  generalised common causes whose effects are expected to be uncorrelated. Nonetheless, the generalised common cause model is demonstrably immune from some of the most common objections to the standard interpretation of the common cause principle. 

For one thing, it has been objected that the screening-off conditions (\ref{equation: c-fork 1})-(\ref{equation: c-fork 2}) are too restrictive, either because they are only satisfied by deterministic common causes \cite{vanfraassen 1980,cartwright 1999}, or because they exclusively apply when the effects of a common cause are independently produced \cite{salmon 1984,cartwright 1988}. This objection is easily met by the generalised common cause model, which drops (\ref{equation: c-fork 1})-(\ref{equation: c-fork 2}) in favour of the more general constraints (\ref{equation: gc-fork 1})-(\ref{equation: gc-fork 2}). For another thing, it has been contended that non-causal positive correlations that result from logical, mathematical, semantic, or  nomic relations do not generally admit of a conjunctive common cause \cite{arntzenius 1992,williamson 2005}. The existence of similar correlations is clearly detrimental to the common undertanding of the principle of the common cause, but it is perfectly consistent with its extended version. The reason is that, according to the extended principle, similar correlations simply do not \emph{call} for a common cause explanation:  by hypothesis, they are determined by logical, mathematical, semantic, or physical laws, so they must be expected, and as such they fall outside its proper domain of application. They are, accordingly, no counterexample to it. 

The following sections will be dedicated to further enrich the generalised common cause model, so as to cover systems of multiple common causes. 


\section{Generalised HR-Reichenbachian Common Cause Systems}\label{generalised hr-reichenbachian common cause systems}


The first attempt to extend the conjunctive common cause model to comprise systems of multiple common causes was made by Hofer-Szab\'{o} and R\'{e}dei in \cite{hofer-szabo redei 2004}, who proposed, to this end, the following definition:
\begin{definition}\label{definition: rccs}
Let $(\Omega, p)$ be a classical probability space with $\sigma$-algebra of random events $\Omega$ and probability measure $p$. For any $A,B\in \Omega$, a \emph{Reichenbachian Common Cause System (HR-RCCS)} of size $n\geq 2$ for $\corr{A}{B}$ is a partition $\set{i = 1}{n}{C_i}$ of $\Omega$ such that:
\begin{eqnarray}
\prob{C_i} \neq 0								 				&	&		(i = 1, ..., n)					\label{equation: rccs 0} \\	
\ccorr{A}{B}{C_{i}} = 0							        			 	&	& 		(i = 1, ..., n)					\label{equation: rccs 1} \\
\big[\cprob{A}{C_i} - \cprob{A}{C_j}\big]\big[ \cprob{B}{C_i} - \cprob{B}{C_j}\big ]> 0     	&	&		(1, ..., n = i\neq j = 1, ..., n).		\label{equation: rccs 2} 
\end{eqnarray}
\end{definition}
Hofer-Szab\'{o} and R\'{e}dei refer to Reichenbachian Common Cause Systems using the acronym RCCS. The acronym HR-RCCS is here employed to distinguish their model from the one utilized in the next section. 

The notion of a HR-RCCS is meant to generalise the notion of a conjunctive common cause in two respects. On the one hand, Hofer-Szab\'{o} and R\'{e}dei demonstrate that only positively correlated pairs admit of a HR-RCCS. On the other hand, conditions (\ref{equation: rccs 0}), (\ref{equation: rccs 1}) and (\ref{equation: rccs 2}) are intended to generalise, respectively, conditions (\ref{equation: c-fork 0})-(\ref{equation: c-fork 0 complement}), (\ref{equation: c-fork 1})--(\ref{equation: c-fork 2}), and (\ref{equation: c-fork 3})--(\ref{equation: c-fork 4}) from Definition \ref{definition: conjunctive fork}. Specifically, (\ref{equation: rccs 1}) demands that each element of a HR-RCCS screen-off its common effects from one another. This means that HR-RCCSs increase the correlation between otherwise uncorrelated pairs, emulating as a consequence the explanatory function of conjunctive common causes.  

Modifying the above definition in accordance with the extended interpretation of common cause principle only requires replacing the screening-off condition (\ref{equation: rccs 1}) with a suitably generalised variant of (\ref{equation: gc-fork 1})-(\ref{equation: gc-fork 2}). Let us accordingly define:
\begin{definition}\label{definition: ghrrccs} 
Let $(\Omega, p)$ be a classical probability space with $\sigma$-algebra of random events $\Omega$ and probability measure $p$. For any $A,B\in \Omega$, a \emph{Generalised Reichenbachian Common Cause System (GHR-RCCS)} of size $n\geq 2$ for $\dev{A}{B}$ is a partition $\set{i = 1}{n}{C_i}$ of $\Omega$ such that:
\begin{eqnarray}
\prob{C_i} \neq 0								 				&	&		(i = 1, ..., n)					\label{equation: grccs 0} \\	
\ccorr{A}{B}{C_{i}} = \ecorr{A}{B}							      	   	&	& 		(i = 1, ..., n)					\label{equation: grccs 1} \\
\big[\cprob{A}{C_i} - \cprob{A}{C_j}\big]\big[ \cprob{B}{C_i} - \cprob{B}{C_j}\big ]> 0     	&	&		(1, ..., n = i\neq j = 1, ..., n).		\label{equation: grccs 2} 
\end{eqnarray}
\end{definition}

This definition generalises at once Definition \ref{definition: generalized conjunctive fork} and Definition \ref{definition: rccs}: it extends the former by admitting systems of \emph{any number} of common causes; it extends the latter by requiring that every common cause in a system should restore the expected correlation between its two effects \emph{whatever its value}. 

Not surprisingly, every GHR-RCCS  increases the correlation between its effects, consequently emulating the explanatory function of generalised common causes. 
\begin{proposition}\label{proposition: grccs positive deviation}
Let $(\Omega, p)$ be a classical probability space with $\sigma$-algebra of random events $\Omega$ and probability measure $p$. For any $A,B\in\Omega$ and any $\set{i\in I}{}{C_i}\subseteq \Omega$, if $\set{i\in I}{}{C_i}$ is a GHR-RCCS of size $n\geq 2$ for $\dev{A}{B}$, then (\ref{equation: positive deviation}) obtains.
\end{proposition}

To show this, let us first prove the following lemma: 
\begin{lemma}\label{lemma: grccs deviation}
Let $(\Omega, p)$ be a classical probability space with $\sigma$-algebra of random events $\Omega$ and probability measure $p$. Let $A,B\in\Omega$ and let $\set{i\in I}{}{C_i}$  be a partition of $\Omega$ satisfying conditions (\ref{equation: grccs 0})-(\ref{equation: grccs 1}). Then:
\begin{equation}\label{equation: grccs deviation}
\dev{A}{B} = \frac{1}{2} \dsum{i,j\in I}{}{\prob{C_i}\prob{C_j}[\cprob{A}{C_i}-\cprob{A}{C_j}][\cprob{B}{C_i}-\cprob{B}{C_j}]}.
\end{equation}
\end{lemma}

\begin{proof}
Let $(\Omega, p)$ be a classical probability space with $\sigma$-algebra of random events $\Omega$ and probability measure $p$. Let $A,B\in\Omega$ and let $\set{i\in I}{}{C_i}$  be a partition of $\Omega$ satisfying (\ref{equation: grccs 0}). From the theorem of total probability it follows that:
\setlength{\jot}{8  pt}
\begin{eqnarray}
\corr{A}{B} 	&   =  	& \frac{1}{2} \dsum{i,j = 1}{n}{\prob{C_{i}}\prob{C_{j}}\big[\cprob{A}{C_{i}}-\cprob{A}{C_{j}}\big]\big[\cprob{B}{C_{i}}-\cprob{B}{C_{j}}\big]} + \nonumber\\
		&	&  \frac{1}{2} \Bigg[\dsum{i = 1}{n}{\prob{C_{i}}\ccorr{A}{B}{C_{i}}} + \dsum{j = 1}{n}{\prob{C_{j}}\ccorr{A}{B}{C_{j}}} \Bigg]. 								 \label{equation: correlation general}
\end{eqnarray}
On the other hand, by hypothesis  $\set{i\in I}{}{C_i}$ is a partition of the given probability space, which implies that:
\setlength{\jot}{6  pt}
\begin{gather}
\dsum{i = 1}{n}{\prob{C_{i}}} = 1. 
\end{gather}
Further assuming (\ref{equation: grccs 1}) will therefore produce the following equality:
\setlength{\jot}{6  pt}
\begin{gather}\label{equation: grccs deviation variant}
\corr{A}{B} - \ecorr{A}{B}  =   \frac{1}{2} \dsum{i,j = 1}{n}{\prob{C_{i}}\prob{C_{j}}[\cprob{A}{C_{i}}-\cprob{A}{C_{j}}][\cprob{B}{C_{i}}-\cprob{B}{C_{j}}]},
\end{gather}
which in the light of (\ref{equation: deviation}) is but a different formulation of (\ref{equation: grccs deviation}).
\end{proof}

Demonstrating Proposition \ref{proposition: grccs positive deviation} on this basis would be straightforward, so we are omitting the details of the proof. One interesting thing to notice about this demonstration, however, is that setting $\ecorr{A}{B} = 0$ in (\ref{equation: grccs deviation variant}) would reduce it to the equation employed by Hofer-Szab\'{o} and R\'{e}dei to demonstrate that HR-RCCSs produce positive correlations. This fact, in itself, is further confirmation of the adequacy of GHR-RCCSs as a generalisation of HR-RCCSs.


\subsection{Existence of GHR-RCCSs}\label{existence of generalised hr-reichenbachian common cause systems}


Hofer-Szab\'{o} and R\'{e}dei \cite{hofer-szabo redei 2006} argue that a HR-RCCS of arbitrary finite size exists for every positively correlated pair of events, in some suitable extension of the original probability space. The discussion to follow will be dedicated to establish a similar result for GHR-RCCSs. Remarkably, it will turn out that not all positive deviations admit of a GHR-RCCS.

Hofer-Szab\'{o} and R\'{e}dei's proof proceeds by noticing that, in general, set $\set{i=1}{n}{C_i}$ is a HR-RCCS of size $n\geq 2$ for $\corr{A}{B}$ in probability space $(\Omega, p)$ if and only if the values of $\cprob{A}{C_{1}}$, ...., $\cprob{A}{C_{n}}$, $\cprob{B}{C_{1}}$, ...., $\cprob{B}{C_{n}}$, and $\prob{C_{1}}$, ...., $\prob{C_{n}}$ satisfy some specified constraints. They call any set $\set{i =1}{n}{a_{i}, b_{i}, c_{i}}$ of $3n$ numbers satisfying said constraints \emph{admissible} for $\corr{A}{B}$ and demonstrate that, for any two positively correlated events $A$ and $B$ and any $n\geq 2$, a set of $n$ admissible numbers for $\corr{A}{B}$ can be found. On this basis, they finally show how an extension of the given probability space can always be constructed, in which some partition $\set{i=1}{n}{C_i}$ exists such that $n\geq 2$ and the values of $\cprob{A}{C_{1}}$, ...., $\cprob{A}{C_{n}}$, $\cprob{B}{C_{1}}$, ...., $\cprob{B}{C_{n}}$, and $\prob{C_{1}}$, ...., $\prob{C_{n}}$ are admissible for $\corr{A}{B}$, thereby etablishing the existence of a HR-RCCS of size $n$ for $\corr{A}{B}$ in that space. 

The following proof will follow the broad logical structure of Hofer-Szab\'{o} and R\'{e}dei's argumentation. Our first step will consist in identifying the necessary and sufficient conditions that must be satisfied by the values of $\cprob{A}{C_{1}}$, ...., $\cprob{A}{C_{n}}$, $\cprob{B}{C_{1}}$, ...., $\cprob{B}{C_{n}}$, $\cprob{A\wedge B}{C_{1}}$, ...., $\cprob{A\wedge B}{C_{n}}$, and $\prob{C_{1}}$, ...., $\prob{C_{n}}$ to make $\set{i=1}{n}{C_i}$ a GHR-RCCS of size $n\geq 2$ for $\dev{A}{B}$. This, however, will be done in two stages, as some of the conditions that we are going to single out will be shared by the model to be developed in \S \ref{generalised m-reichenbachian common cause systems}. Let us begin by isolating these.

\begin{definition}\label{definition:quasi-admissible}
Let $(\Omega, p)$ be a classical probability space with $\sigma$-algebra of random events $\Omega$ and probability measure $p$. For any $A,B\in \Omega$ satisfying (\ref{equation: positive deviation}) and any $n \geq 2$, the set 
\begin{equation*}
\set{i = 1}{n}{a_{i}, b_{i}, c_{i}, d_{i}}
\end{equation*}
 of real numbers is called \emph{quasi-admissible* for} $\dev{A}{B}$  if and only if the following conditions hold:
\begin{gather}
\dsum{i = 1}{n}{a_{i}c_{i}} = \prob{A} 																																						\label{equation: adm a}		\\
\dsum{i = 1}{n}{b_{i}c_{i}} = \prob{B}																																							\label{equation: adm b}		\\
\dsum{i =1}{n}{c_{i}} = 1																																									\label{equation: adm partition}	\\
d_{i}- a_{i}b_{i}  = \ecorr{A}{B}																																\phantom{12 pt}	(i = 1, ..., n)				\label{equation: adm di}		\\
0 <a_{i}, b_{i}, d_{i} < 1  																																\phantom{12 pt}	(i = 1, ..., n)				\label{equation: adm aibi}	\\
0 < c_{i} < 1 																																		\phantom{12 pt}	(i = 1, ..., n).			\label{equation: adm ci}	 	
\end{gather}
\end{definition}
The attentive reader will have noticed that, for each $n\geq 2$, quasi-admissible sets include $4n$ numbers, whereas admissible sets, as defined by Hofer-Szab\'{o} and R\'{e}dei, include only $3n$ numbers. Moreover, while (\ref{equation: adm a})-(\ref{equation: adm partition}) and (\ref{equation: adm aibi})-(\ref{equation: adm ci}) are either identical to or straightforward generalizations of some of  Hofer-Szab\'{o} and R\'{e}dei's original conditions for admissible numbers, constraint (\ref{equation: adm partition}) is not. Similar changes are needed to avoid a logical mistake in their original proof, along the lines illustrated in more detail in \cite{mazzola evans 2017}.

To complete this part of the proof, we need to supplement quasi-admissible sets with one more condition:
\begin{definition}\label{definition: grccs admissible}
Let $(\Omega, p)$ be a classical probability space with $\sigma$-algebra of random events $\Omega$ and probability measure $p$. For any $A,B\in \Omega$ and any $n \geq 2$, a set $\set{i = 1}{n}{a_{i}, b_{i}, c_{i}, d_{i}}$ of real numbers  is called \emph{HR-admissible} for $\dev{A}{B}$ if and only if it is quasi-admissible for $\dev{A}{B}$ and it further satisfies
\begin{gather}\label{equation: grccs admissible}
[a_{i}-a_{j}][b_{i}-b_{j}]> 0    	\phantom{12 pt} (1, ..., n = i \neq j = 1, ..., n).	
\end{gather}
\end{definition}

The adequacy of the above definition is testified by the following lemma, whose proof is straightforward and which can consequently be omitted:
\begin{lemma}\label{lemma: grccs admissible}
Let $(\Omega, p)$ be a classical probability space with $\sigma$-algebra of random events $\Omega$ and probability measure $p$. For any $A,B\in \Omega$ and any $\set{i = 1}{n}{C_{i}}\subseteq\Omega$ where $n\geq 2$, the set $\set{i = 1}{n}{C_{i}}$ is a GHR-RCCS of size $n$ for $\dev{A}{B}$ if and only if there exists a set $\set{i = 1}{n}{a_{i}, b_{i}, c_{i}, d_{i}}$ of HR-admissible numbers for $\dev{A}{B}$ such that  
\begin{gather}
\prob{C_{i}} = c_{i} 									\phantom{12 pt}	(i = 1, ..., n)					\label{equation: adm 1} \\
\cprob{A}{C_{i}} = a_{i} 									\phantom{12 pt}	(i = 1, ..., n)					\label{equation: adm 2} \\
\cprob{B}{C_{i}} = b_{i} 									\phantom{12 pt}	(i = 1, ..., n)					\label{equation: adm 3} \\
\cprob{A\wedge B}{C_{i}} = d_{i}							\phantom{12 pt}	(i = 1, ..., n).				\label{equation: adm 4}
\end{gather}
\end{lemma}

The next step in our proof will be to establish the necessary and sufficient conditions for the existence of HR-admissible numbers for $\dev{A}{B}$. To this purpose, however, we shall need the following lemma:
\begin{lemma}\label{lemma: quasi-admissible variant}
Let $(\Omega, p)$ be a classical probability space with $\sigma$-algebra of random events $\Omega$ and probability measure $p$. Moreover, let $A,B \in \Omega$ and let $\set{i = 1}{n}{C_{i}}\subseteq\Omega$ with $n\geq 2$. Then, any set $\set{i = 1}{n}{a_{i}, b_{i}, c_{i}, d_{i}}$ of real numbers satisfying identities (\ref{equation: adm 1})-(\ref{equation: adm 4}) is quasi-admissible for $\dev{A}{B}$ if and only if it further satisfies (\ref{equation: adm aibi})-(\ref{equation: adm ci}) as well as
\begin{equation}
a_{n} = \frac{a-\dsum{k=1}{n-1}{c_{k} a_{k}}}{1-\dsum{k=1}{n-1}{c_{k}}}																																	\label{equation: an}		
\end{equation}
\begin{equation}
b_{n} = \frac{b-\dsum{k=1}{n-1}{c_{k} b_{k}}}{1-\dsum{k=1}{n-1}{c_{k}}}																																	\label{equation: bn}		
\end{equation}
\begin{equation}
c_{n} = 1- \dsum{k=1}{n-1}{c_{k}} 																																							\label{equation: cn}		
\end{equation}
\begin{equation}
d_{n} = \varepsilon + \frac{\left[a - \dsum{k =1}{n-1}{a_{k}c_{k}}\right]\left[b - \dsum{k = 1}{n-1}{b_{k}c_{k}}\right]}{\left[ 1 - \dsum{k = 1}{n-1}{c_{k}} \right ]^{2}} 																			\label{equation: dn}		
\end{equation}
\begin{equation}
d_{k} = \varepsilon +  a_{k}b_{k} 																															\phantom{12 pt}	(k = 1, ..., n-1),			\label{equation: di}		
\end{equation}
where
\setlength{\jot}{6  pt}
\begin{gather}
a = \prob{A}																		\label{equation: a}			\\
b = \prob{B} 																	\label{equation: b}			\\
\varepsilon = \ecorr{A}{B}.															\label{equation: epsilon}
\end{gather}
\end{lemma}

\begin{proof}
Let $(\Omega, p)$ be a classical probability space with $\sigma$-algebra of random events $\Omega$ and probability measure $p$. Moreover,  let $A,B\in \Omega$ satisfy (\ref{equation: positive deviation}) and let the set $\set{i = 1}{n}{a_{i}, b_{i}, c_{i}, d_{i}}$ of $n\geq 2$ real numbes  satisfy conditions  (\ref{equation: adm 1})-(\ref{equation: adm 4}) and (\ref{equation: adm aibi})-(\ref{equation: adm ci}). Finally, let (\ref{equation: a})-(\ref{equation: epsilon}) be in place. 

Given the aforesaid hypothesis,  (\ref{equation: adm a})-(\ref{equation: adm partition}) can be directly obtained from (\ref{equation: an})-(\ref{equation: cn}) thanks to the theorem of total probability, and vice-versa. Therefore, we only need to show that (\ref{equation: adm di}) obtains if and only if (\ref{equation: dn})-(\ref{equation: di}) do. To this purpose, let us first observe that, as a further consequence of the theorem of total probability, the following equality holds: 
\begin{equation}\label{equation: dn general}
d_{n} =   [d_{n} - a_{n}b_{n}] + \frac{\left[a - \dsum{k = i}{n-1}{a_{k}c_{k}}\right]\left[b - \dsum{k = 1}{n-1}{b - b_{k}}\right]}{\left[ 1 - \dsum{k = 1}{n-1}{c_{k}}\right ]^{2}}.																									
\end{equation}
Thanks to (\ref{equation: dn general}) it is then immediate to verify that (\ref{equation: dn})-(\ref{equation: di}) are simultaneously satisfied if (\ref{equation: adm di}) is. Conversely, let us suppose that (\ref{equation: dn})-(\ref{equation: di}) are the case. Then (\ref{equation: dn}) and (\ref{equation: dn general}) will jointly imply that 
\begin{gather}
d_{n} -  a_{n}b_{n} = \varepsilon, 
\end{gather}
which, together with (\ref{equation: di}), straighforwardly implies (\ref{equation: adm di}), as required.
\end{proof}

Endowed with the above result, we are now in a position to determine the necessary conditions so that, in general, HR-admissible numbers $\set{i = 1}{n}{a_{i}, b{i}, c_{i}, d_{i}}$ could exist for $\dev{A}{B}$ and $n\geq 2$. Quite interestingly:

\begin{lemma}\label{lemma: grccs admissible constraint}
Let $(\Omega, p)$ be a classical probability space with $\sigma$-algebra of random events $\Omega$ and probability measure $p$. For any $A,B\in \Omega$, no HR-admissible numbers exist for $\dev{A}{B}$ if 
\begin{gather}\label{equation: grccs admissible constraint}
\ecorr{A}{B} + \prob{A}\prob{B} \leq 0.
\end{gather}
\end{lemma}

\begin{proof}
Let $(\Omega, p)$ be a classical probability space with $\sigma$-algebra of random events $\Omega$ and probability measure $p$, and let $A, B\in \Omega$ be arbitrarily chosen. Lemma \ref{lemma: grccs admissible constraint} will be established by contraposition, so let us assume that, for some $n\geq 2$, a set $\set{i = 1}{n}{a_{i}, b_{i}, c_{i}, d_{i}}$ of HR-admissible numbers does exist for $\dev{A}{B}$. Moreover, let us assume identities (\ref{equation: a})-(\ref{equation: epsilon}). 

To prove our lemma, two preliminary steps will be required. First, we shall prove that some $a_{j}, b_{k}\in \set{i = 1}{n}{a_{i}, b_{i}, c_{i}, d_{i}}$ exist such that 
\begin{gather}
a - a_{j} >  0 	\label{equation: grccs preempter of a}\\
b - b_{k} >  0.	\label{equation: grccs preempter of b} 
\end{gather}
Next, on that basis, we shall demonstrate that at least some such $a_{j}, a_{k}\in \set{i = 1}{n}{a_{i}, b_{i}, c_{i}, d_{i}}$ exist, for which $j = k$.

To estabish the first claim, let us begin by noticing that, as a plain consequence of (\ref{equation: adm a})-(\ref{equation: adm b}) and (\ref{equation: adm partition}):
\begin{gather}
0 = a - a =\dsum{i = 1}{n}{ac_{i}} - \dsum{i = 1}{n}{a_{i}c_{i}} = \dsum{i = 1}{n}{c_{i}(a - a_{i})} \\
0 = b - b =\dsum{i = 1}{n}{bc_{i}} - \dsum{i = 1}{n}{b_{i}c_{i}} = \dsum{i = 1}{n}{c_{i}(b - b_{i})}.
\end{gather}
On the other hand, (\ref{equation: adm ci}) demands that $c_{i} > 0$ for all $i = 1, ..., n$, while (\ref{equation: grccs admissible}) implies that $a_{j} = a$ and $b_{k} = b$ can be satisfied by at most one term $a_{j}$ and one term $b_{k}$ for $j,k = 1, ...,n\geq 2$. The above equalities therefore imply that $a- a_{i}$ be positive for some values of $i$ and negative for others, while similarly $b- b_{i}$ be positive for some values of $i$ and negative for others. This is enough to prove  (\ref{equation: grccs preempter of a}) and (\ref{equation: grccs preempter of b}), as desired.


To prove our second auxiliary result, let us first relabel all numbers in $\set{i = 1}{n}{a_{i}, b_{i}, c_{i}, d_{i}}$ so that 
\begin{gather}
a_{1} < ... <  a_{k} < a \leq a_{k+1} < ... < a_{n}.
\end{gather}
This in turn implies that
\begin{gather}\label{equation: grccs reshuffle variant}
a_{i} - a_{j} < 0 \phantom{12 pt} i = 1, ..., k; \phantom{6 pt} j = k+1, ..., n.  
\end{gather}
Now, let us proceed by reductio, and let us assume that 
\begin{gather}
b_{i} - b > 0 \phantom{12 pt} i = 1, ..., k.
\end{gather}
Then, according to the result previously established, some $b_{j}\in \set{i= k+1}{n}{b_{i}}\subset \set{i = 1}{n}{a_{i}, b_{i}, c_{i}, d_{i}}$ should exist such that 
\begin{gather}
b - b_{j} > 0. 	
\end{gather}
However, in that case 
\begin{gather}
b_{i} - b_{j} > 0 \phantom{12 pt} i = 1, ..., k
\end{gather}
would ensue. Together with (\ref{equation: grccs reshuffle variant}), this would imply
\begin{gather}
[a_{i} - a_{j}][b_{i} - b_{j}] < 0 \phantom{12 pt} i = 1, ..., k
\end{gather}
consequently contradicting (\ref{equation: grccs admissible}). By reductio, this shows that  (\ref{equation: grccs preempter of a})-(\ref{equation: grccs preempter of b}) must be satisfied for some $a_{j}, b_{k}\in \set{i = 1}{n}{a_{i}, b_{i}, c_{i}, d_{i}}$ where $j = k$.

Let us now come to the main part of our proof. Thanks to the results so established, we can now safely claim that, for any set $\set{i = 1}{n}{a_{i}, b_{i}, c_{i}, d_{i}}$ of HR-admissible numbers, some $a_{i},b_{i}\in \set{i = 1}{n}{a_{i}, b_{i}, c_{i}, d_{i}}$ are always to be found such that 
\begin{gather}
ab - a_{i}b_{i} > 0.
\end{gather}
Together with (\ref{equation: adm di}), this implies that
\begin{gather}
ab + \varepsilon > a_{i}b_{i} + \varepsilon = d_{i} > 0,
\end{gather}
contradicting (\ref{equation: grccs admissible constraint}). By contraposition, this means that whenever (\ref{equation: grccs admissible constraint}) is satisfied, no set $\set{i = 1}{n}{a_{i}, b_{i}, c_{i}, d_{i}}$ of real numbers can satify (\ref{equation: adm di}) given the other conditions for a HR-admissible set for $\dev{A}{B}$. Hence, no HR-admissible set can exist for $\dev{A}{B}$. 
\end{proof}

Let us now move to the sufficient condition for the existence of HR-admissible numbers $\set{i = 1}{n}{a_{i}, b_{i}, c_{i}, d_{i}}$ for $\dev{A}{B}$ and $n\geq 2$. Remarkably, this turns out to be the same as the necessary condition:
\begin{lemma}\label{lemma: grccs-admissible numbers exist}
Let $(\Omega, p)$ be a classical probability space with $\sigma$-algebra of random events $\Omega$ and probability measure $p$. For any $A,B\in \Omega$ satisfying (\ref{equation: positive deviation}), a set $\set{i = 1}{n}{a_{i}, b_{i}, c_{i}, d_{i}}$ of HR-admissible numbers for $\dev{A}{B}$ exists for each $n\geq 2$ if
\begin{gather}\label{equation: admissible grccs admissible}
\ecorr{A}{B} + \prob{A}\prob{B} > 0.
\end{gather}
\end{lemma}

\begin{proof}
Let $(\Omega, p)$ be a classical probability space with $\sigma$-algebra of random events $\Omega$ and probability measure $p$. Moreover, let $A,B\in \Omega$ satisfy (\ref{equation: positive deviation}) and (\ref{equation: admissible grccs admissible}).

To start with, let us observe that (\ref{equation: di}) is in fact a system of $n-1$ equations, namely one for each value of $i = 1, ..., n-1$. Therefore,  (\ref{equation: an})-(\ref{equation: di}) jointly comprise a system of $4 + (n - 1) = n + 3$ equations in $4 n$ variables. This means that each HR-admissible set for $\dev{A}{B}$ is determined by a set of $4n - (n+3) = 3n- 3$ parameters, for every $n\geq 2$. To establish the existence of such set, we accordingly need to prove that such parameters exist. To this purpose, let numbers $a$, $b$ and $\varepsilon$ be understood as per (\ref{equation: a})-(\ref{equation: epsilon}). Proof will proceed by induction on the cardinality of $n$. 

Let us begin by assuming $n = 2$ as our inductive basis. This has the effect of transforming (\ref{equation: an})-(\ref{equation: di}), respectively, into:
\setlength{\jot}{8  pt}
\begin{gather}
a_{2} = \frac{a - c_{1} a_{1}}{1-c_{1}}																																			\label{equation: a2}		\\ 
b_{2} = \frac{b- c_{1}  b_{1}}{1-c_{1}}																																			\label{equation: b2}		\\		
c_{2} = 1- c_{1}																																							\label{equation: c2}		\\ 
d_{2} = \varepsilon + \frac{[a - a_{1}c_{1}][b - b_{1}c_{1}]}{[ 1 - c_{1}]^{2}} 																													\label{equation: d2}		\\
d_{1} -  a_{1}b_{1} = \varepsilon. 																																				\label{equation: d1}		     											\end{gather}	
Since $a$, $b$ and $\varepsilon$ are known by hypothesis, choosing numbers $c_{1}$, $a_{1}$ and $b_{1}$ will therefore suffice to fix the values of all $4n = 8$ variables in the system. Let us accordingly constrain $c_{1}$ so that:
\setlength{\jot}{6  pt}
\begin{gather}
c_{i} \rightarrow 0.
\end{gather}
Owing to this, (\ref{equation: a2})-(\ref{equation: d2}) immediately produce
\begin{gather}
a_{2} \rightarrow a \\
b_{2} \rightarrow b \\
c_{2} \rightarrow 1 \\
d_{2} \rightarrow \varepsilon+ ab,																																				\label{equation: d2 limit}
\end{gather}
while on the other hand (\ref{equation: positive deviation}) directly requires that 
\begin{gather} 
1 > a,b > 0, 			\label{equation: ab constraint}
\end{gather}
as it would be easy to verify. Taken together, this ensures that
\begin{gather}
1 > a_{2}, b_{2} > 0 																																						\label{equation: a1a2b1b2}	\\
1 >  c_{1}, c_{2} > 0,																																						\label{equation: c1c2}	 	
\end{gather}
while (\ref{equation: admissible grccs admissible}) and (\ref{equation: positive deviation}) imply that
\begin{gather}
1 > d_{2} > 0.
\end{gather}
To determine the remaining numbers, we further need to set $a_{1}$ and $b_{1}$. In this case, our choice will depend on the value of $\varepsilon$, as follows:
\setlength{\jot}{12  pt}
\begin{gather}
\varepsilon \geq 0 	\phantom{12 pt} \begin{cases} a > a_{1} > 0 \\	b > b_{1} > 0 \end{cases} \label{equation: a1b1 i} \\
\varepsilon < 0	\phantom{12 pt} \begin{cases} 1 > a_{1} > 0 \\	1 > b_{1} > 0 \end{cases} \label{equation: a1b1 ii}
\end{gather}
Either option is allowed by (\ref{equation: ab constraint}), and either will ensure that
\begin{gather}
1 > d_{i} > 0,
\end{gather}
as it would be straightforward to check with the aid of (\ref{equation: di}), (\ref{equation: positive deviation}) and (\ref{equation: admissible grccs admissible}). Thanks to Lemma (\ref{lemma: quasi-admissible variant}), this is enough to establish that some set $\set{i = 1}{n}{a_{i}, b_{i}, c_{i}, d_{i}}$ of quasi-admissible numbers exist for $\dev{A}{B}$ if $n = 2$. To further show that such set is HR-admissible for $\dev{A}{B}$, we only need to observe that (\ref{equation: grccs admissible}) can be obtained from both (\ref{equation: a1b1 i}) and (\ref{equation: a1b1 ii}), owing to (\ref{equation: adm a})-(\ref{equation: adm partition}).

Let us now assume, as our inductive hypothesis, that some set $\set{i = 1}{m}{a_{i}, b_{i}, c_{i}, d_{i}}$ be HR-admissible for $\dev{A}{B}$, where $n = m > 2$. To prove that a HR-admissible set for $\dev{A}{B}$ also exists if $n = m +1$, let us consider the set 
\begin{equation*}
\set{j = 1}{m-2}{a_{j}, b{j}, c_{j}, a_{m-1}, b_{m-1}}\subset\set{i = 1}{m}{a_{i}, b_{i}, c_{i}, d_{i}},
\end{equation*}
and let us choose numbers $a'_{m}  , b'_{m}, c'_{m-1}, c'_{m} $ such that: 
\begin{gather}
a_{j} > a'_{m}  > 0 			\phantom{12 pt}	( j = 1, ..., m-1)	\label{equation: a'm} \\
b_{j} > b'_{m}  > 0 			\phantom{12 pt}	( j = 1, ..., m-1)	\label{equation: b'm} \\
c_{m}, c_{m-1} > c'_{m-1} > 0 								\label{equation: c'm-1}\\
1 > c'_{m} > 0.				 						\label{equation: c'm}
\end{gather}
Given (\ref{equation: an})-(\ref{equation: di}), the set
\begin{gather*}
\set{j = 1}{m-2}{a_{j}, b{j}, c_{j}, a_{m-1}, b_{m-1}}\cup\set{}{}{a'_{m}, b'_{m}, c'_{m-1}, c'_{m}} 
\end{gather*}
of $3(m-1)+3 = 3(m+1)-3 = 3n-3$ parameters will then suffice to determine $4(m+1)$ numbers:
\begin{gather*}
\set{j = 1}{m-2}{a_{j}, b_{j}, c_{j}, d_{j}, a_{m-1}, b_{m-1}, c'_{m-1}, d_{m-1}, a'_{m}, b'_{m}, c'_{m}, d'_{m}, a_{m+1}, b_{m+1}, c_{m+1}, d_{m+1}}.
\end{gather*}
Because numbers $\set{j = 1}{m-1}{a_{j}, b_{j}, c_{j}, d_{j}, a_{m-1}, b_{m-1}, c'_{m-1}, d_{m-1}, a'_{m}, b'_{m}, c'_{m},}$ satisfy conditions (\ref{equation: adm di})-(\ref{equation: adm ci}) and (\ref{equation: grccs admissible}) by hypothesis, all we need to show is that said constraints be also satisfied by the remaining numbers $\set{}{}{d'_{m}, a_{m+1}, b_{m+1}, c_{m+1},  d_{m+1}}$. To this purpose, let us first notice that  (\ref{equation: adm aibi}) must be true of $d'_{m}$ by virtue of (\ref{equation: adm di}) and (\ref{equation: a'm})-(\ref{equation: b'm}). Next, thanks to (\ref{equation: an})-(\ref{equation: dn}),  it will be sufficient to suppose that
\begin{gather}
a'_{m} \rightarrow 0 \\
b'_{m} \rightarrow 0 \\
c'_{m} \rightarrow c_{m-1}-c'_{m-1}
\end{gather}
to obtain
\begin{gather}
a_{m+1} \rightarrow a_{m}\\
b_{m+1} \rightarrow b_{m}\\
c_{m+1} \rightarrow c_{m}\\
d_{m+1} \rightarrow d_{m}
\end{gather}
which we already know, by our inductive hypothesis, to satisfy (\ref{equation: adm di})-(\ref{equation: adm ci}). Moreover, (\ref{equation: a'm}) and (\ref{equation: b'm}), along with the inductive assumption whereby
\begin{gather}
[a_{m} - a_{i}][b_{m}-b_{i}] > 0 \phantom{12 pt} (i = 1, ..., m-1),
\end{gather}
ensures that 
\begin{gather}
[a_{m+1} - a_{i}][b_{m+1}-b_{i}] > 0 \phantom{12 pt} (i = 1, ..., m),
\end{gather}
which together with our inductive hypothesis suffices to establish (\ref{equation: grccs admissible}).  Due to Lemma \ref{lemma: quasi-admissible variant} and Definition \ref{definition: grccs admissible}, the set of $4(m+1)$ numbers so determined is therefore HR-admissible for $\dev{A}{B}$. Lemma \ref{lemma: grccs-admissible numbers exist} is thus demonstrated by induction.  
\end{proof}

Before getting to the end of our existential proof, we need one more definition:
\begin{definition}\label{definition: extension}
Let $(\Omega, p)$ and $(\Omega', p')$ be  classical probability spaces with $\sigma$-algebras of random events $\Omega$ and $\Omega'$ and with probability measures $p$ and $p'$, respectively. Then $(\Omega', p')$ is called an \emph{extension} of $(\Omega, p)$ if and only if there exists an injective lattice homomorphism $h: \Omega \rightarrow \Omega'$, preserving complementation, such that 
\begin{equation}
\altprob{h(X)} = \prob{X}	\phantom{12 pt} \mbox{for all } X\in \Omega.
\end{equation}
\end{definition}

The results of our demonstration can thus be crystallized into the following proposition:
\begin{proposition}\label{proposition: existence of grccs}
Let $(\Omega, p)$ be a classical probability space with $\sigma$-algebra of random events $\Omega$ and probability measure $p$. For any $A,B\in \Omega$ satisfying (\ref{equation: positive deviation}) and any $n \geq 2$, an extension $(\Omega', p')$ of $(\Omega, p)$ including a GHR-RCCS of size $n$ for $\dev{A}{B}$ exists if and only if $A$ and $B$ satisfy (\ref{equation: admissible grccs admissible}).
\end{proposition}

\begin{proof}
The only-if clause immediately follows from Lemma \ref{lemma: grccs admissible} and Lemma \ref{lemma: grccs admissible constraint}. Proof of the if-clause, instead, is structurally similar to Step 2 of Hofer-Szab\'{o} and R\'{e}dei's proof for the existence of HR-RCCSs of arbitrary finite size in \cite{hofer-szabo redei 2006}, although  Hofer-Szab\'{o} and R\'{e}dei's conditions (60)-(63) will have to be replaced by:
\begin{gather}
r_{i}^{1} = \frac{c_{i}d_{i}}{\prob{A\wedge B}}										\label{equation: r1i} 	\\
r_{i}^{2} = \frac{c_{i}a_{i}-c_{i}d_{i}}{\prob{A\wedge \overline{B}}} 							\label{equation: r2i}	\\
r_{i}^{3} = \frac{c_{i}b_{i}-c_{i}d_{i}}{\prob{\overline{A}\wedge B}} 							\label{equation: r3i}	\\
r_{i}^{4} = \frac{c_{i} -c_{i}a_{i}-c_{i}b_{i}+c_{i}d_{i}}{\prob{\overline{A}\wedge \overline{B}}},		\label{equation: r4i}	
\end{gather}
which, owing to (\ref{equation: adm di}), actually reduce to the aforesaid conditions for $\varepsilon = 0$. 
\end{proof}


\section{Generalised Reichenbachian Common Cause Systems Revisited}\label{generalised m-reichenbachian common cause systems}


There are two aspects in which HR-RCCSs may not be considered fully satisfactory generalisations of conjunctive common causes. The first aspect is that they can admit of elements that are probabilistically independent of one or both events from the corresponding correlated pair. This is at odds with the intuition that positive causes should \textit{ceteris paribus} increase the probability of their effects, and that negative causes should \textit{ceteris paribus} decrease their probability. The second aspect is that they rule out the possibility that two distinct causes could equally alter the probability of one, or both, of their effects. On the face of it, there is simply no reason why a systems of common causes should be so constrained. To overcome these limitations, in \cite{mazzola 2012} I proposed a revisitation of HR-RCCSs, along the following lines:
\begin{definition}\label{definition: mccs}
Let $(\Omega, p)$ be a classical probability space with $\sigma$-algebra of random events $\Omega$ and probability measure $p$. For any $A,B\in \Omega$, a \emph{Reichenbachian Common Cause System* (M-RCCS)} of size $n\geq 2$ for $\corr{A}{B}$ is a partition $\set{i = 1}{n}{C_{i}}$ of $\Omega$ such that:
\begin{eqnarray}
\prob{C_i} \neq 0								 				&	&		(i = 1, ..., n)					\label{equation: mccs 0} \\	
\ccorr{A}{B}{C_{i}} = 0							        			 	&	& 		(i = 1, ..., n)					\label{equation: mccs 1} \\
\big[\cprob{A}{C_i} - \prob{A}\big]\big[ \cprob{B}{C_i} - \prob{B}\big ]> 0     			&	&		(i = 1, ..., n).				\label{equation: mccs 2} 
\end{eqnarray}
\end{definition}

Whether M-RCCSs are really to be preferred to HR-RCCSs is controversial \cite{stergiou 2015}. However, this is no place to settle that issue. Rather, in this section we shall limit ourselves to offer an alternative extension of the generalised common cause model, by taking M-RCCSs as a basis. Let us accordingly define:
\begin{definition}\label{definition: gmccs}
Let $(\Omega, p)$ be a classical probability space with $\sigma$-algebra of random events $\Omega$ and probability measure $p$. For any $A,B\in \Omega$, a \emph{Generalised Reichenbachian Common Cause System* (GM-RCCS)} of size $n\geq 2$ for $\corr{A}{B}$ is a partition $\set{i = 1}{n}{C_{i}}$ of $\Omega$ such that:
\begin{eqnarray}
\prob{C_i} \neq 0								 				&	&		(i = 1, ..., n)					\label{equation: gmccs 0} \\	
\ccorr{A}{B}{C_{i}} = \ecorr{A}{B}							        		&	& 		(i = 1, ..., n)					\label{equation: gmccs 1} \\
\big[\cprob{A}{C_i} - \prob{A}\big]\big[ \cprob{B}{C_i} - \prob{B}\big ]> 0     			&	&		(i = 1, ..., n ).				\label{equation: gmccs 2} 
\end{eqnarray}
\end{definition}

Just as with GHR-RCCSs, it can be shown that GM-RCCSs invariably produce a positive deviation between the observed correlation of their effects and their expected correlation. To this end, let us first introduce the following lemma:
\begin{lemma}\label{lemma: gmccs deviation}
Let $(\Omega, p)$ be a classical probability space with $\sigma$-algebra of random events $\Omega$ and probability measure $p$. Let $A,B\in\Omega$ and let $\set{i\in I}{}{C_i}$  be a partition of $\Omega$ satisfying conditions (\ref{equation: gmccs 0})-(\ref{equation: gmccs 1}). Then:
\begin{equation}\label{equation: gmccs deviation}
\dev{A}{B} = \dsum{i\in I}{}{\prob{C_i}[\cprob{A}{C_i}-\prob{A}][\cprob{B}{C_i}-\prob{B}]}.
\end{equation}
\end{lemma}

\begin{proof}
Let $(\Omega, p)$ be a classical probability space with $\sigma$-algebra of random events $\Omega$ and probability measure $p$. Let $A,B\in\Omega$ and let $\set{i\in I}{}{C_i}$ be a partition of $\Omega$ for which (\ref{equation: gmccs 0}) holds. The theorem of total probability thereby implies that
\begin{gather}
\corr{A}{B} = \dsum{i = 1}{n}{\prob{C_{i}}[\cprob{A}{C_{i}}-\prob{A}][\cprob{B}{C_{i}}-\prob{A}]}  + \dsum{i = 1}{n}{\prob{C_{i}}\ccorr{A}{B}{C_{i}}}.
\end{gather}
Let us now suppose that (\ref{equation: gmccs 1}) be satisfied, too.  Then, owing to the fact that 
\begin{gather}
	\dsum{i = 1}{n}{\prob{C_{i}}} = 1,
\end{gather}
few elementary calculations would transform the above equality into:
 \begin{gather}
\corr{A}{B} - \ecorr{A}{B} = \dsum{i = 1}{n}{\prob{C_{i}}[\cprob{A}{C_{i}}-\prob{A}][\cprob{B}{C_{i}}-\prob{A}]},  
\end{gather}
which according to (\ref{equation: deviation}) is just a restatement of (\ref{equation: gmccs deviation}).
\end{proof}
Based on the above lemma, it would then be easy to demonstrate the following proposition:
\begin{proposition}\label{proposition: grccs positive deviation}
Let $(\Omega, p)$ be a classical probability space with $\sigma$-algebra of random events $\Omega$ and probability measure $p$. For any $A,B\in\Omega$ and any $\set{i\in I}{}{C_i}\subseteq \Omega$, if $\set{i\in I}{}{C_i}$ is a GM-RCCS of size $n\geq 2$ for $\dev{A}{B}$, then (\ref{equation: positive deviation}) obtains.
\end{proposition}

GM-RCCSs accordingly perform a similar explanatory function as GHR-RCCSs. Quite interestingly, moreover, for every two events $A$ and $B$ in a classical probability space such that $\dev{A}{B}>0$ and every $n\geq 2$, a GM-RCCSs of size $n$ for $\dev{A}{B}$ exists in some extension of the given probability space  if and only if a GHR-RCCS does.  Demonstrating this will be our next objective. 


\subsection{Existence of GM-RCCSs}\label{existence of generalised m-reichenbachian common cause systems}


The existential proof we shall elaborate in this section will follow the broad lines of the one developed in \S \ref{existence of generalised hr-reichenbachian common cause systems}. Just as with GHR-RCCSs, we shall first determine the necessary and sufficient conditions the values of $\cprob{A}{C_{1}}$, ...., $\cprob{A}{C_{n}}$, $\cprob{B}{C_{1}}$, ...., $\cprob{B}{C_{n}}$, $\cprob{A\wedge B}{C_{1}}$, ...., $\cprob{A\wedge B}{C_{n}}$, and $\prob{C_{1}}$, ...., $\prob{C_{n}}$ ought to satisfy so that $\set{i = 1}{n}{C_{i}}$ be a GM-RCCS of size $n\geq 2$ for $\dev{A}{B}$. Subsequently, we shall determine the necessary and sufficient conditions for the existence of such numbers, and on that basis we shall finally establish the necessary and sufficient conditions for the existence of an extension of the given probability space, where a GM-RCCS of size $n$ for $\dev{A}{B}$ could be found. 

Quite evidently, GM-RCCSs differ from GHR-RCCSs only in that they substitute condition (\ref{equation: grccs 2}) with (\ref{equation: gmccs 2}). Consequently, in order to complete the first step of our proof, we only need to replace Definition \ref{definition: grccs admissible} with the following one:

\begin{definition}\label{definition: gmccs admissible}
Let $(\Omega, p)$ be a classical probability space with $\sigma$-algebra of random events $\Omega$ and probability measure $p$. For any $A,B\in \Omega$ and any $n \geq 2$, a set $\set{i = 1}{n}{a_{i}, b_{i}, c_{i}, d_{i}}$ of real numbers is called \emph{M-admissible} for $\dev{A}{B}$ if and only if it is quasi-admissible for $\dev{A}{B}$ and it further satisfies
\begin{gather}\label{equation: gmccs admissible}
[a_{i}-\prob{A}][b_{i}-\prob{B}]> 0    	\phantom{12 pt} (i = 1, ..., n ).	
\end{gather}
\end{definition}

Just as before, the adequacy of the above definition is easily established:
\begin{lemma}\label{lemma: gmccs admissible}
Let $(\Omega, p)$ be a classical probability space with $\sigma$-algebra of random events $\Omega$ and probability measure $p$. For any $A,B\in \Omega$ and any $\set{i = 1}{n}{C_{i}}\subseteq\Omega$ where $n\geq 2$, the set $\set{i = 1}{n}{C_{i}}$ is a GM-RCCS of size $n$ for $\dev{A}{B}$ if and only if there exists a set $\set{i = 1}{n}{a_{i}, b_{i}, c_{i}, d_{i}}$ of M-admissible numbers for $\dev{A}{B}$ for which identities (\ref{equation: adm 1})-(\ref{equation: adm 4}) are true. 
\end{lemma}

M-admissible numbers are quasi-admissible by definition. This fact allows us to build on our previous discussion, to easily prove the following result:
\begin{lemma}\label{lemma: gmccs admissible constraint}
Let $(\Omega, p)$ be a classical probability space with $\sigma$-algebra of random events $\Omega$ and probability measure $p$. For any $A,B\in \Omega$, no M-admissible numbers exist for $\dev{A}{B}$ if 
\begin{gather}\label{equation: gmccs admissible constraint}
\ecorr{A}{B} + \prob{A}\prob{B} \leq 0. \tag{\ref{equation: grccs admissible constraint}}
\end{gather}
\end{lemma}

\begin{proof}
Lemma \ref{lemma: gmccs admissible constraint} is demonstrated in a similar way as Lemma \ref{lemma: grccs admissible constraint}. Let us accordingly $(\Omega, p)$ be a classical probability space with $\sigma$-algebra of random events $\Omega$ and probability measure $p$, let $A, B\in \Omega$ be arbitrarily chosen so as to satisfy (\ref{equation: positive deviation}), and let us further assume that for some $n\geq 2$, a set $\set{i = 1}{n}{a_{i}, b_{i}, c_{i}, d_{i}}$ of M-admissible numbers exists for $\dev{A}{B}$. Moreover, let  identities (\ref{equation: adm 1})-(\ref{equation: adm 4}) and(\ref{equation: a})-(\ref{equation: epsilon}) be in place.

Showing that some $a_{i}, b_{i}\in \set{i = 1}{n}{a_{i}, b_{i}, c_{i}, d_{i}}$ exist satisfying 
\begin{gather}
a-a_{i}> 0\\
b-b_{i} > 0
\end{gather}
in this case would only require some elementary calculations, as the above inequality directly follow from (\ref{equation: gmccs admissible}) along with (\ref{equation: adm a}) and (\ref{equation: adm b}). The remainder of the proof would then proceed in exactly the same way as the analogous proof for Lemma \ref{lemma: grccs admissible constraint}.
\end{proof}

Condition (\ref{equation: admissible grccs admissible}) is thus necessary for the existence of M-admissible numbers for $\dev{A}{B}$, for any two events $A$ and $B$ satisfying (\ref{equation: positive deviation}) and any $n\geq 2$. Moreover, as with GHR-RCCSs, it is also sufficient: 
\begin{lemma}\label{lemma: gmccs-admissible numbers exist}
Let $(\Omega, p)$ be a classical probability space with $\sigma$-algebra of random events $\Omega$ and probability measure $p$. For any $A,B\in \Omega$ satisfying (\ref{equation: positive deviation}), a set $\set{i = 1}{n}{a_{i}, b_{i}, c_{i}, d_{i}}$ of M-admissible numbers for $\dev{A}{B}$ exists for each $n\geq 2$ if
\begin{gather}\label{equation: admissible gmccs admissible}
\ecorr{A}{B} + \prob{A}\prob{B} > 0. \tag{\ref{equation: admissible grccs admissible}}
\end{gather}
\end{lemma}

\begin{proof}
Let $(\Omega, p)$ be a classical probability space with $\sigma$-algebra of random events $\Omega$ and probability measure $p$. Moreover, let $A,B\in \Omega$ satisfy (\ref{equation: positive deviation}) and (\ref{equation: admissible gmccs admissible}). Proof will proceed by induction on $n$.

Let $n = 2$ accordingly be our inductive basis. Because for $n = 2$ conditions (\ref{equation: grccs admissible}) and (\ref{equation: gmccs admissible}) become equivalent, this case was actually covered in the inductive proof for Lemma \ref{lemma: grccs-admissible numbers exist}. Next, as our inductive hypothesis, let $n = m$ and let $\set{i = 1}{m}{a_{i}, b_{i}, c_{i}, d_{i}}$ be M-admissible for $\dev{A}{B}$. On this basis, let us now proceed to the last step of our inductive proof, and let $n = m+(r-1)$, where $r\geq 2$.

Let us first choose some $c_{k}\in \set{i = 1}{m}{a_{i}, b_{i}, c_{i}, d_{i}}$. Then, (\ref{equation: adm partition}) and (\ref{equation: adm ci}) ensure that it is possible to find a set $\set{j = 1}{r}{c_{k}^{j}}$ of $r\geq 2$ identical real numbers, lying inside the interval $(0,1)$, such that 
\begin{equation}
\dsum{j = 1}{r}{c_{k}^{j}} =\dsum{j = 1}{r}{\frac{c_{k}}{r}} =  c_{k}. 
\end{equation}
Furthermore, it is trivially possible to find three sets $\set{j = 1}{r}{a_{k}^{j}}$, $\set{j = 1}{r}{b_{k}^{j}}$ and $\set{j = 1}{r}{d_{k}^{j}}$ of $r\geq 2$ identical real numbers satisfying
\begin{gather}
a_{k}^{j} = a_{k} \phantom{12 pt} j = 1, ..., r \\
b_{k}^{j} = b_{k} \phantom{12 pt} j = 1, ..., r \\
d_{k}^{j} = d_{k} \phantom{12 pt} j = 1, ..., r. 
\end{gather} 
Given (\ref{equation: a})-(\ref{equation: b}), our inductive hypothesis then implies:
\setlength{\jot}{8  pt}
\begin{gather}
a = 
			\dsum{i=1}{m}{a_{i}c_{i}} = 
			\dsum{k\neq i = 1}{m}{a_{i}c_{i}} + a_{k}c_{k} = 
			\dsum{k\neq i = 1}{m}{a_{i}c_{i}} + a_{k}\frac{c_{k}}{r}r	= 
			\dsum{k\neq i = 1}{m}{a_{i}c_{i}} + \dsum{j = 1}{r}{a_{k}^{j}c_{k}^{j}} =
			\dsum{k\neq i = 1}{m + r}{a_{i}c_{i}} \\
b = 
			\dsum{i=1}{m}{b_{i}c_{i}} = 
			\dsum{k\neq i = 1}{m}{b_{i}c_{i}} + b_{k}c_{k} = 
			\dsum{k\neq i = 1}{m}{b_{i}c_{i}} + b_{k}\frac{c_{k}}{r}r	= 
			\dsum{k\neq i = 1}{m}{b_{i}c_{i}} + \dsum{j = 1}{r}{b_{k}^{j}c_{k}^{j}} =
			\dsum{k\neq i = 1}{m + r}{b_{i}c_{i}}  \\
1 = 			
			\dsum{i=1}{m}{c_{i}} = 
			\dsum{k\neq i = 1}{m}{c_{i}} + c_{k} = 
			\dsum{k\neq i = 1}{m}{c_{i}} + \dsum{j = 1}{r}{c_{k}^{j}} =
			\dsum{k\neq i = 1}{m + r}{c_{i}}  \\
\varepsilon =	 
			d_{k} - a_{k}b_{k} = d_{k}^{j} - a_{k}^{j}b_{k}^{j} \phantom{12 pt} j = 1, ..., r.
\end{gather}
This guarantees that the set $\set{k\neq i = 1}{m}{a_{i}, b_{i}, c_{i}, d_{i}}\cup\set{j = 1}{r}{a_{j}, b_{j}, c_{j}, d_{j}}$ of $4(m+r-1)$ numbers so obtained satisfies (\ref{equation: adm a})-(\ref{equation: adm di}). Furthermore, (\ref{equation: adm aibi})-(\ref{equation: adm ci}) and (\ref{equation: gmccs admissible}) are clearly satisfied owing to our inductive hypothesis and to the way numbers $\set{j = 1}{r}{a_{j}, b_{j}, c_{j}, d_{j}}$ were chosen. This is enough to prove that  $\set{k\neq i = 1}{m}{a_{i}, b_{i}, c_{i}, d_{i}}\cup\set{j = 1}{r}{a_{j}, b_{j}, c_{j}, d_{j}}$ is M-admissible for $\dev{A}{B}$, therefore concluding our inductive proof.
\end{proof}

Our existential proof is now virtually complete. Let us just add one final touch:
\begin{proposition}\label{proposition: existence of gmccs}
Let $(\Omega, p)$ be a classical probability space with $\sigma$-algebra of random events $\Omega$ and probability measure $p$. For any $A,B\in \Omega$ satisfying (\ref{equation: positive deviation}) and any $n \geq 2$, an extension $(\Omega', p')$ of $(\Omega, p)$ including a GM-RCCS of size $n$ for $\dev{A}{B}$ exists if and only if $A$ and $B$ satisfy (\ref{equation: admissible grccs admissible}).
\end{proposition}

\begin{proof}
Proof is in all similar to the proof for Proposition \ref{proposition: existence of grccs}, mutatis mutandis.
\end{proof}

Two final remarks may be added at this at this point. First,  Proposition \ref{proposition: existence of grccs} and Proposition \ref{proposition: existence of gmccs} both rectify the results announced in \cite{mazzola 2013}, where it was implicitly assumed that expected correlations be greater than or equal to zero, leading to the erroneous claim that a generalised common cause should exist, in some extension of the initial probability space, for every positive deviation. Second, GHR-RCCSs and GM-RCCSs for a given deviation may not coexist in the same probability space. Neverthless, Proposition \ref{proposition: existence of grccs} and \ref{proposition: existence of gmccs} jointly guarantee the following result:
\begin{proposition}\label{proposition: existence of grccs and gmccs}
Let $(\Omega, p)$ be a classical probability space with $\sigma$-algebra of random events $\Omega$ and probability measure $p$. For any $A,B\in \Omega$ satisfying (\ref{equation: positive deviation}), an extension $(\Omega', p')$ of $(\Omega, p)$ including a GHR-RCCS of size $n\geq 2$ for $\dev{A}{B}$ exists if and only if there is some extension $(\Omega'', p'')$ of $(\Omega, p)$  including a GM-RCCS of size $m\geq 2$ for $\dev{A}{B}$. 
\end{proposition}

This means that, irrespective of the different probabilistric properties of GHR-RCCSs and GM-RCCSs, neither model can explain more or different deviations than the other. The two models are thus to be assessed based not on what they can explain, but how. The question as to whether M-RCCSs should be preferred to HR-RCCSs, therefore, is carried over to their generalised counterparts.


\subsection{Conclusion}


The principle of the common cause decrees that improbable coincidences be put down to the action of some common cause. The standard interpretation of the principle takes this as a requirement that positive correlations between causally unrelated events be removed by conditioning on some conjunctive common cause. The interpretation here promoted, and encapsulated in the extended principle of the common cause, urges by contrast that common causes be called for in order to explain positive deviations between the estimated correlation of two events and their expected correlation. This paper has outlined two distinct probabilistic models for systems of common causes that incorporate the extended intepretation of the principle. GHR-RCCSs have been elaborated by combining the generalised common cause model with HR-RCCSs. GM-RCCSs, instead, have been obtained by integrating generalised common causes with M-RCCSs. The necessary and sufficient conditions for the existence of finite systems of either kind have been determined. 

Our demonstration led to the unexpected result that some extension of the given classical probability space can be found including a GHR-RCCS of arbitrary finite size for some specified positive deviation, if and only if a similar exstension can be found including a GM-RCCS of finite size for the same deviation. Even more interestingly, in either case the existence of such space is guaranteed if and only if the sum of the expected correlation of the pair of events under consideration and the product of their probabilities is greater than zero. The mathematical reason for this limitation is clear: only under said constraint, in fact, can HR-admissible numbers and M-admissible numbers for a positive deviation exist. The philosophical interpretation of this result, instead, is an open question. 



\end{document}